\documentclass[10pt,a4paper]{amsart}
\usepackage[margin=20mm]{geometry}
\usepackage[numbers]{natbib}
\usepackage{hyperref}

\usepackage{amssymb,amsmath}
\usepackage{amsthm}
\usepackage{bbm}
\usepackage{stmaryrd}
\usepackage[usenames,dvipsnames]{xcolor}
\usepackage{color}

\input xy
\xyoption{all} %\tolerance=500

\numberwithin{equation}{section}

\newtheorem{theorem}{Theorem}[section]

\newtheorem{lemma}[theorem]{Lemma}
\newtheorem{proposition}[theorem]{Proposition}

\theoremstyle{definition}
\newtheorem{definition}[theorem]{Definition}
\newtheorem{remark}[theorem]{Remark}
\newtheorem{example}[theorem]{Example}

\newcommand{\Id}{\mathbbmss{1}}

\newcommand{\Bsigma}{\mbox{\boldmath$\sigma$}}

\newcommand{\InHom}{\mbox{$\underline{\Hom}$}}

\newcommand{\rmd}{\textnormal{d}}
\newcommand{\rme}{\textnormal{e}}

\DeclareMathOperator{\Vect}{Vect}

\DeclareMathOperator{\Ber}{Ber}

\DeclareMathOperator{\Hom}{Hom}

\DeclareMathOperator{\tr}{tr}

\font\black=cmbx10 \font\sblack=cmbx7 \font\ssblack=cmbx5 \font\blackital=cmmib10  \skewchar\blackital='177
\font\sblackital=cmmib7 \skewchar\sblackital='177 \font\ssblackital=cmmib5 \skewchar\ssblackital='177
\font\sanss=cmss10 \font\ssanss=cmss8 %scaled 900
\font\sssanss=cmss8 scaled 600 \font\blackboard=msbm10 \font\sblackboard=msbm7 \font\ssblackboard=msbm5
\font\caligr=eusm10 \font\scaligr=eusm7 \font\sscaligr=eusm5  \font\fraktur=eufm10
\font\sfraktur=eufm7 \font\ssfraktur=eufm5 
\font\bsymb=cmsy10 scaled\magstep2
\def\all#1{\setbox0=\hbox{\lower1.5pt\hbox{\bsymb
       \char"38}}\setbox1=\hbox{$_{#1}$} \box0\lower2pt\box1\;}
\def\exi#1{\setbox0=\hbox{\lower1.5pt\hbox{\bsymb \char"39}}
       \setbox1=\hbox{$_{#1}$} \box0\lower2pt\box1\;}

\def\tx#1{{\fam0\relax#1}}

\newfam\bifam
\textfont\bifam=\blackital \scriptfont\bifam=\sblackital \scriptscriptfont\bifam=\ssblackital

\newfam\blfam
\textfont\blfam=\black \scriptfont\blfam=\sblack \scriptscriptfont\blfam=\ssblack

\newfam\bbfam
\textfont\bbfam=\blackboard \scriptfont\bbfam=\sblackboard \scriptscriptfont\bbfam=\ssblackboard

\newfam\ssfam
\textfont\ssfam=\sanss \scriptfont\ssfam=\ssanss \scriptscriptfont\ssfam=\sssanss
\def\sss#1{{\fam\ssfam\relax#1}}

\newfam\clfam
\textfont\clfam=\caligr \scriptfont\clfam=\scaligr \scriptscriptfont\clfam=\sscaligr

\newfam\frfam
\textfont\frfam=\fraktur \scriptfont\frfam=\sfraktur \scriptscriptfont\frfam=\ssfraktur

\def\hpb#1{\setbox0=\hbox{${#1}$}
    \copy0 \kern-\wd0 \kern.2pt \box0}
\def\vpb#1{\setbox0=\hbox{${#1}$}
    \copy0 \kern-\wd0 \raise.08pt \box0}

\def\pmb#1{\setbox0\hbox{${#1}$} \copy0 \kern-\wd0 \kern.2pt \box0}
\def\pmbb#1{\setbox0\hbox{${#1}$} \copy0 \kern-\wd0
      \kern.2pt \copy0 \kern-\wd0 \kern.2pt \box0}
\def\pmbbb#1{\setbox0\hbox{${#1}$} \copy0 \kern-\wd0
      \kern.2pt \copy0 \kern-\wd0 \kern.2pt
    \copy0 \kern-\wd0 \kern.2pt \box0}
\def\pmxb#1{\setbox0\hbox{${#1}$} \copy0 \kern-\wd0
      \kern.2pt \copy0 \kern-\wd0 \kern.2pt
      \copy0 \kern-\wd0 \kern.2pt \copy0 \kern-\wd0 \kern.2pt \box0}
\def\pmxbb#1{\setbox0\hbox{${#1}$} \copy0 \kern-\wd0 \kern.2pt
      \copy0 \kern-\wd0 \kern.2pt
      \copy0 \kern-\wd0 \kern.2pt \copy0 \kern-\wd0 \kern.2pt
      \copy0 \kern-\wd0 \kern.2pt \box0}

\mathchardef\za="710B  %\alpha
\mathchardef\zb="710C  %\beta
\mathchardef\zg="710D  %\gamma
\mathchardef\zd="710E  %\delta
\mathchardef\zve="710F %\epsilon
\mathchardef\zz="7110  %\zeta
\mathchardef\zh="7111  %\eta
\mathchardef\zvy="7112 %\theta
\mathchardef\zi="7113  %\iota
\mathchardef\zk="7114  %\kappa
\mathchardef\zl="7115  %\lambda
\mathchardef\zm="7116  %\mu
\mathchardef\zn="7117  %\nu
\mathchardef\zx="7118  %\xi
\mathchardef\zp="7119  %\pi
\mathchardef\zr="711A  %\rho
\mathchardef\zs="711B  %\sigma
\mathchardef\zt="711C  %\tau
\mathchardef\zu="711D  %\upsilon
\mathchardef\zvf="711E %\phi
\mathchardef\zq="711F  %\chi
\mathchardef\zc="7120  %\psi
\mathchardef\zw="7121  %\omega
\mathchardef\ze="7122  %\varepsilon
\mathchardef\zy="7123  %\vartheta
\mathchardef\zf="7124  %\varomega
\mathchardef\zvr="7125 %\varrho
\mathchardef\zvs="7126 %\varsigma
\mathchardef\zf="7127  %\varphi
\mathchardef\zG="7000  %\Gamma
\mathchardef\zD="7001  %\Delta
\mathchardef\zY="7002  %\Theta
\mathchardef\zL="7003  %\Lambda
\mathchardef\zX="7004  %\Xi
\mathchardef\zP="7005  %\Pi
\mathchardef\zS="7006  %\Sigma
\mathchardef\zU="7007  %\Upsilon
\mathchardef\zF="7008  %\Phi
\mathchardef\zW="700A  %\Omega
\mathchardef\zC="7009  %\Psi

\newcommand{\be}{\begin{equation}}
\newcommand{\ee}{\end{equation}}

\newcommand{\bea}{\begin{eqnarray}}
\newcommand{\eea}{\end{eqnarray}}
\def\*{{\textstyle *}}
\newcommand{\R}{{\mathbb R}}

\newcommand{\Z}{{\mathbb Z}}

\newcommand{\s}{{\textstyle *}}

\def\Hom{\sss{Hom}}

\def\Vect{\sss{Vect}}

\def\xi{\tx{i}}

\def\cM{\cal M}

\def\s*{{\scriptstyle *}}

\def\cO{\mathcal{O}}

\def\cL{\mathcal{L}}

\def\cM{\mathcal{M}}

\newcommand{\beas}{\begin{eqnarray*}}
\newcommand{\eeas}{\end{eqnarray*}}
\def\half{\frac{1}{2}}

\title{$\Z_2\times \Z_2$-graded supersymmetry:  $2$-$\textnormal{d}$ sigma models}

   \author{Andrew James Bruce} 
   \address{Mathematics Research Unit, University of Luxembourg, Maison du Nombre 6, avenue de la Fonte, 
L-4364 Esch-sur-Alzette}  
   \email{andrewjamesbruce@googlemail.com}

\date{\today}
 
\begin{document}

\begin{abstract}
We propose a natural $\Z_2 \times \Z_2$-graded generalisation of $d=2$, $\mathcal{N}=(1,1)$  supersymmetry and construct a $\Z_2^2$-space realisation thereof. Due to the grading, the supercharges close with respect to, in the classical language, a commutator rather than an anticommutator.  This is then used to build  classical (linear and non-linear) sigma models  that exhibit this novel supersymmetry via mimicking standard superspace methods.  The  fields in our models are  bosons,  right-handed and left-handed Majorana-Weyl spinors, and exotic bosons. The bosons commute with all the fields, the spinors belong to different sectors that cross commute rather than anticommute,  while the exotic boson anticommute with the spinors. As a particular example of one of the models, we present a `double-graded' version of supersymmetric sine-Gordon theory. \par
\end{abstract}

 \maketitle

\setcounter{tocdepth}{2}
 \tableofcontents
 
\section{Introduction and Preliminaries} 
\subsection{Introduction}
Inspired by the recent developments in both $\Z_2^n$-geometry (see \cite{Bruce:2019a,Bruce:2019b,Bruce:2019c,Bruce:2020,Covolo:2016,Covolo:2016b}) and the appearance of $\Z_2^n$-gradings in theoretical physics (see \cite{Aizawa:2020a,Aizawa:2020b,Aizawa:2020c,Aizawa:2020d,Bruce:2019aa,Bruce:2020a,Tolstoy:2014a}), we propose a natural generalisation the $d=2$, $\mathcal{N}= (1,1)$ supersymmetry algebra that is $\Z_2\times \Z_2$-graded, or more colloquially, double-graded.  Taking this new $\Z_2^2$-Lie algebra ($\Z_2^2 := \Z_2\times \Z_2$) as our starting point, we develop  `superspace' methods to allow us to define  double-graded versions of  two-dimensional supersymmetric sigma models. Sigma models have a long history, starting with Gell-Mann \& L\'{e}vy (see \cite{GellMann:1960}), and have provided many interesting links between field theory and differential geometry.  Good reviews of supersymmetric sigma models can be found in \cite{Deligne:1999,Freed:1999}. For an overview of supersymmetry and related topics, the reader my consult the encyclopedia edited by Duplij, Siegel \&  Bagger  \cite{Duplij:2004}. We also remark that two-dimensional supersymmetry algebras play a fundamental r\^ole in superstring theory. Importantly, under quite general assumptions, two-dimensional field theories are renormalisable, including those with highly non-linear interactions. We will not touch upon quantisation in this paper. \par 
Part of the motivation for this work was to extend the classical models that have $\Z_2^2$-supersymmetry as first defined by the author (in four dimensions) in \cite{Bruce:2019a}.  This gauntlet was picked up by Aizawa, Kuznetsova \& Toppan who showed that there is a plethora of mechanical models, both classical and quantum that do indeed exhibit this kind of supersymmetry (see \cite{Aizawa:2020c,Aizawa:2020d}).  The first quantum mechanical model, a direct generalisation of Witten's model \cite{Witten:1981}, was given by the author and Duplij (see \cite{Bruce:2020a}).  Given that the spin-statistics correspondence is not necessarily true in less than three spatial dimensions, it is plausible that this novel $\Z_2^2$-supersymmetry could be realised in experimental condensed matter physics (see for example \cite{MacKenzie:1988}). Thus, getting a handle on classical and quantum field theories in low dimensions that are $\Z_2^2$-supersymmetric is of potential importance in  physics. The models we construct here are $1+1$-dimensional, and we must mention that many  systems  have effectively  one spatial dimension, such as quantum wires, carbon nanotubes, and edge states in quantum Hall systems and topological insulators. Non-linear sigma models have long been applied to condensed matter physics. Supersymmetry has also been applied in condensed matter physics, notably starting with the pioneering  works of  Efetov on the statistical properties of energy levels and eigenfunctions in disordered systems (see \cite{Efetov:1983}).\par 
As the field theories we present are inherently two-dimensional, it is particularly convenient to use Dirac's light-cone coordinates (see \cite{Dirac:1949}), and this is reflected in our initial definitions.   Not only  does this choice of coordinates shorten the length of some expression as compared to writing them in inertial coordinates, but it also makes the transformation properties under two-dimensional Lorentz boosts of various expressions clear. The basic component fields in our models are  bosons, right-handed Majorana-Weyl spinors,  left-handed Majorana-Weyl spinors, and exotic bosons.  Typically, the right-handed and left-handed spinors belong to different sectors that cross commute while the exotic bosons anticommute with the spinors. The models will consist of dynamical bosons and fermions together with non-propagating exotic bosons, or  dynamical exotic bosons and fermions together with  non-propagating boson.  The former being more natural from our perspective of sigma models. We refer to models with propagating exotic bosons as ``exotic models''.  The reader should, of course, be reminded of Green--Volkov parastatistics (see \cite{Green:1953,Volkov:1959}). The r\^{o}le of $\Z_2^n$-gradings in parastatistics and parasupersymmetry has long been recognised (see, for example, \cite{Tolstoy:2014,Yang:2001,Yang:2001a}). However, as shown in \cite{Bruce:2020a}, the $\Z_2^2$-Lie algebras we study are not the same as those found in parasupersymmetry.  \par 

\medskip

\noindent \textbf{Arrangement.} For the remainder of this section, we recall the notion of a $\Z_2^2$-Lie algebra and remind the reader of the basics of $\Z_2^2$-geometry as needed  for the main part of the paper.  In Section \ref{Sec:Z22Mink} we define the $d=2$ $\mathcal{N}= (1,1)$ $\Z_2^2$-supertranslation  algebra (Definition \ref{Def:Z22SuperAlg}) and then present a representation of this algebra on a $\Z_2^2$-manifold (see Definition \ref{Def:Z22Mink}). That is, we construct the $\Z_2^2$-supercharges, etc. as vector fields on a  $\Z_2^2$-graded version of two-dimensional super-Minkowski spacetime.  We build some aspects of ``$\Z_2^2$-space methods'' that we then apply to construct sigma models that are $\Z_2^2$-supersymmetric (see Definition \ref{Def:SigmaModel} and Definition \ref{Def:NonLinSigmaModel}) in Section \ref{Sec:SigmaModels}.   We end the main text with a few closing remarks in Section \ref{Sec:ConRem}.  Two appendices are included, Appendix \ref{App:Ber} recalls the definition of the $\Z_2^2$-Berezinian and Appendix \ref{App:BerInt} on Berezin integration on $\Z_2^n$-space with two degree zero coordinates and one coordinate each of the non-zero degrees. The general theory of integration on $\Z_2^n$-domains, or even just $\Z_2^2$-domains with more that one of each non-zero degree coordinate is, at the time of writing this paper, work in progress.  Appendix \ref{App:BerInt} is included to prove that the integration method used in this paper is mathematically sound.

\subsection{$\Z_2^2$-Lie algebras}
Here we  recall the notion of a $\Z_2^2$-Lie algebra (see \cite{Rittenberg:1978,Scheunert:1979}).   A $\Z_2^2$-graded vector space is a vector space (over $\mathbb{R}$ or $\mathbb{C}$) that is the direct sum of homogeneous vector spaces
$$\mathfrak{g} = \mathfrak{g}_{00} \oplus \mathfrak{g}_{11}   \oplus \mathfrak{g}_{01} \oplus \mathfrak{g}_{10}\,.$$
Note that we have fixed an ordering for the elements of $\Z_2^2$ and that other orderings do appear in the literature.  We will denote the $\Z_2^2$-degree of  a  homogeneous element $ a \in \mathfrak{g}$ as $\mathrm{deg}(a) \in \mathbb{Z}_2^2$.  The scalar product on $\mathbb{Z}_2^2$ is inherited from the standard scalar product on  $\R^2$ and  we denote this by $\langle -, - \rangle$. That is, if $\mathrm{deg}(a) = (\gamma_1, \gamma_2)$ and $\mathrm{deg}(b) = (\gamma^\prime_1, \gamma^\prime_2)$, then $\langle \mathrm{deg}(a), \mathrm{deg}(b) \rangle = \gamma_1 \gamma^\prime_1 +  \gamma_2 \gamma^\prime_2$. A $\Z_2^2$-graded vector space had a decomposition into its even and odd subspaces, defined by the total degree,
\begin{align*}
&\mathfrak{g}_{ev} := \mathfrak{g}_{00} \oplus \mathfrak{g}_{11}, && \mathfrak{g}_{od} := \mathfrak{g}_{01} \oplus \mathfrak{g}_{10}.
\end{align*}
\begin{definition}
A \emph{$\Z_2^2$-Lie algebra} is a  $\Z_2^2$-graded vector space equipped with a bi-linear operation, $[-,-]$, such that for homogeneous elements  $a,b$ and $c\in \mathfrak{g}$, the following are satisfied:\\
\begin{enumerate}
\itemsep1em 
\item $\mathrm{deg}([a,b]) = \mathrm{deg}(a) + \mathrm{deg}(b)$,
\item $[a,b] = {-} (-1)^{\langle \mathrm{deg}(a), \mathrm{deg}(b) \rangle} [b,a]$,
\item $[a,[b,c] ] = [[a,b],c] = + (-1)^{\langle \mathrm{deg}(a), \mathrm{deg}(b) \rangle} [b, [a,c]]$.
\end{enumerate}
\end{definition}
Extension to inhomogeneous elements is via linearity.  We have written the graded Jacobi identity for a $\mathbb{Z}_2^2$-Lie algebra in Loday-Leibniz form, though due to the graded antisymmetry of the Lie bracket one can recast this in a more traditional form. Moreover, generalising this definition to $\Z_2^n$-Lie algebras is straightforward.

\subsection{Elements of $\Z_2^n$-geometry}
The locally ringed space approach  to $\Z_2^n$-manifolds  was pioneered by Covolo, Grabowski and Poncin (see \cite{Covolo:2016}) We restrict  attention to real $\Z_2^n$-manifolds and will not consider the complex analogues.\par 
\begin{definition}[\cite{Covolo:2016}]
A \emph{locally} $\Z_{2}^{n}$-\emph{ringed space}, $n \in \mathbb{N} \setminus \{0\}$, is a pair $S := (|S|, \mathcal{O}_{S} )$ where $|S|$ is a second-countable Hausdorff topological space and a $\cO_{S}$  is a sheaf  of $\Z_{2}^{n}$-graded, $\Z_{2}^{n}$-commutative associative unital $\mathbb{R}$-algebras, such that the stalks $\cO_{S,p}$, $p \in  |S|$ are local rings.
\end{definition}
Here,  $\Z_{2}^{n}$-commutative means that any two sections $a$, $b \in \mathcal{O}_S(|U|)$, $|U| \subset |S|$ open, of homogeneous degree $\deg(a)$ and $\deg(b) \in \Z_{2}^{n}$, respectively, commute with a Koszul sign rule defined by the standard scalar product
$$ab = (-1)^{\langle \deg(a), \deg(b)\rangle} \: ba.$$
We need to fix a convention on  the order of elements in $\Z_{2}^{n}$, we do this  filling in zeros from the left and ones from the right and then putting the elements with zero total degree at the front while keeping their relative order.  For example, and pertinent for this paper, we order the elements of $\Z_2^2$ as
$$
 \Z_{2}^{2} := \Z_2 \times \Z_2 =   \big \{ (0,0),  \: (1,1), \: (0,1), \: (1,0) \big\},$$
which, of course, agrees with how we ordered the homogeneous subspaces of a $\Z_2^2$-graded vector space. A tuple $\mathbf{q} = (q_{1}, q_{2}, \cdots , q_{N})$, where  $N = 2^{n}-1$ provides all the information about the non-zero degree coordinates, which we collectively write as $\zx$. 
\begin{definition}[\cite{Covolo:2016}]
A (smooth) $\Z_{2}^{n}$-\emph{manifold} of dimension $p |\mathbf{q}$ is a locally $\Z_{2}^{n}$-ringed space $ M := \left(|M|, \mathcal{O}_{M} \right)$, which is locally isomorphic to the $\Z_{2}^{n}$-ringed space $\mathbb{R}^{p |\mathbf{q}} := \left( \mathbb{R}^{p}, C^{\infty}_{\R^p}[[\zx]] \right)$. Here   $C^{\infty}_{\R^p}$ is the structure sheaf on the Euclidean space $\R^p$.  Local sections of $\mathbb{R}^{p |\mathbf{q}}$ are formal power series in the $\Z_{2}^{n}$-graded variables $\zx$ with  smooth coefficients, i.e.,
$$   C^{\infty}(\mathcal{U}^p)[[\zx]] =  \left \{ \sum_{\alpha \in \mathbb{N}^{N}} \zx^{\alpha}f_{\alpha} ~ | \: f_{\alpha} \in C^{\infty}(\mathcal{U}^p)\right \},$$
for any  $\mathcal{U}^p \subset \R^p$ open.  \emph{Morphisms} between $\Z_{2}^{n}$-manifolds are  morphisms of $\Z_{2}^{n}$-ringed spaces, that is,  pairs $(\phi, \phi^{*}) : (|M|, \mathcal{O}_{M}) \rightarrow  (|N|, \mathcal{O}_{N})$ consisting of a continuous map  $\phi: |M| \rightarrow |N|$ and sheaf morphism $\phi^{*}: \mathcal{O}_{N}(|V|) \rightarrow \mathcal{O}_{M}(\phi^{-1}(|V|))$, where $|V| \subset |N|$  is open.
\end{definition}
Many of the standard results from the theory of supermanifolds generalise to $\Z_{2}^{n}$-manifolds (see \cite{Manin;1997} for an introduction to the theory of supermanifolds). For example, the topological space $|M|$ comes with the structure of a smooth manifold of dimension $p$.  Let $M$ be a $\Z_{2}^{n}$-manifold, then an \emph{open $\Z_{2}^{n}$-submanifold} of $M$ is a  $\Z_{2}^{n}$-manifold defined as 
$$U := \big( |U|, \cO_{M,|U|} \big),$$
where $|U| \subseteq |M|$ is an open subset. From the definition of a $\Z_2^n$-manifold, we know that for `small enough' $|U|$ we have an isomorphism
$$ (\varphi, \varphi^*): U  \stackrel{\sim}{\rightarrow} \mathcal{U}^{p|\mathbf{q}}\,.$$
This local diffeomorphism allows the construction of  a local coordinate system.  We write the local coordinates  as $x^{A} = (x^\mu, \zx^i)$.  The commutation rules for these coordinates are given
$$x^{A}x^{B} = (-1)^{\langle \deg(A), \deg(B)\rangle}\:x^{B}x^{A}\,.$$
Changes of coordinates, i.e., different choices of the local isomorphisms,  can be written (using standard abuses of notation) as $x^{A'} = x^{A'}(x)$, where we understand the changes of coordinates to respect the $\Z_{2}^{n}$-grading. Note that generically we have a formal power series rather than just polynomials. We will refer to global sections of the structure sheaf of a $\Z_{2}^{n}$-manifold as \emph{functions} and employ the standard notation  $ C^{\infty}(M):=\cO_M(|M|)$. 
\begin{example}[$\Z_2^n $-graded  Cartesian spaces]\label{Exp:Z2nCartSp} Directly from the definition,   $\mathbb{R}^{p |\mathbf{q}} := \left( \mathbb{R}^{p}, C^{\infty}(\mathbb{R}^{p})[[\zx]] \right)$ is a $\Z_2^n $-manifold. Global coordinates $(x^\mu, \zx^i)$ can be employed, where the coordinate map is just the identity. In this paper, we will only meet $\Z_2^2$-manifolds that are globally isomorphic to $\R^{2|1,1,1}$.
\end{example} 
We have a \emph{chart theorem} (\cite[Theorem 7.10]{Covolo:2016}) that allows us to uniquely extend morphisms between the local coordinate domains to morphisms of locally $\Z_{2}^{n}$-ringed spaces. That is, we can describe morphisms using local coordinates. Naturally,  the  $\Z_2^n$-degree of a coordinate must be preserved under the pull-back by a morphism. \par
A vector field on a $\Z_{2}^{n}$-manifold is  a $\Z_{2}^{n}$-derivation of the global functions. Thus, a (homogeneous) vector field is  a linear map $X: C^{\infty}(M) \rightarrow C^{\infty}(M)$, that satisfies the $\Z_2^n$-graded Leibniz rule
$$X(fg) = X(f) g + (-1)^{\langle \deg(X), \deg(f) \rangle}\: f X(g),$$
for any  homogeneous $f$ and not necessarily homogeneous $g \in C^{\infty}(M)$. It is easy to check that the space of vector fields  has a (left) $C^{\infty}(M)$-module structure. We denote the module of vector fields as $\Vect(M)$. An arbitrary  vector field can be `localised' (see \cite[Lemma 2.2]{Covolo:2016b}) in the sense that given $|U| \subset |M|$ there always exists a unique derivation
$$X|_{|U|} : \cO_M(|U|) \rightarrow \cO_M(|U|),$$
such that $X(f)|_{|U|} = X|_{|U|}(f_{|U|})$. Because of this local property, it is clear that one has a sheaf of $\cO_M$-modules formed by the local derivations, this defines the \emph{tangent sheaf} of a $\Z_{2}^{n}$-manifold (see \cite[Definition 5.]{Covolo:2016b}). Moreover, this sheaf is locally free. Thus, we can always write a vector field locally as 
$$X = X^{A}(x)\frac{\partial}{\partial x^{A}},$$
where the partial derivatives are defined as standard for the coordinates of degree zero and are defined algebraically for the remaining coordinates. We drop the explicit reference to the required restriction as is standard in differential geometry. The order of taking partial derivatives matters, but only up to sign factors,
$$\frac{\partial}{\partial x^{A}} \frac{\partial}{\partial x^{B}} =  (-1)^{\langle \deg(A), \deg(B)  \rangle } \: \frac{\partial}{\partial x^{B}} \frac{\partial}{\partial x^{A}}.$$
Under the ($\Z_2^n$-graded) commutator 
$$[X,Y] :=  X \circ Y - (-1)^{\langle \deg(X), \deg(Y)\rangle } \:  Y\circ X,$$
$\Vect(M)$ becomes a $\Z_{2}^{n}$-Lie algebra (see \cite{Covolo:2012,Rittenberg:1978,Scheunert:1979}). The  grading and symmetry  of the Lie bracket are clear, and one can directly check that the Jacobi identity  
$$[X,[Y,Z]] = [[X,Y],Z] + (-1)^{\langle \deg(X), \deg(Y) \rangle } \: [Y, [X,Z]],$$
holds. \par 
There are now many other technical results known about $\Z_2^n$-manifolds, most of which we will not require in the rest of this paper. The interested reader can consult \cite{Bruce:2019b,Bruce:2019c,Bruce:2020}.

\section{Two dimensional $\mathcal{N} =(1|1)$ $\Z_2^2$-Minkowski spacetime}\label{Sec:Z22Mink}
\subsection{The $\Z_2^2$-graded  supertranslation algebra}
Following our earlier work \cite{Bruce:2019a}, we propose the following $\Z_2^2$-Lie algebra as the starting place for our constructions.
\begin{definition}\label{Def:Z22SuperAlg}
The \emph{$\Z_2 \times \Z_2$-graded, $d=2$, $\mathcal{N}= (1,1)$ supertranslation algebra} is the $\Z_2^2$-Lie algebra with $5$ generators with the following assigned degrees
$$\{\underbrace{P_-}_{(0,0)}, ~ \underbrace{P_+}_{(0,0)}, ~ \underbrace{Z_{-+}}_{(1,1)},~\underbrace{Q_-}_{(0,1)} ,~ \underbrace{Q_+}_{(1,0)} \} ,$$
where they are assumed to transform under $2$-d Lorentz boosts (here $\beta$ is the rapidity) as
\begin{align*}
& P_-   \mapsto \rme^\beta \, P_-, 
& P_+   \mapsto \rme^{-\beta} \, P_+,
\\
& Z_{-+} \mapsto Z_{-+}. &\\
& Q_- \mapsto \rme^{\half \beta}\, Q_-,
& Q_+ \mapsto \rme^{-\half \beta}\, Q_+,\\
\end{align*}  
with the  non-vanishing $\Z_2 \times \Z_2$-graded Lie brackets
\begin{align*}
& [Q_-, Q_-] = P_-, && [Q_+, Q_+] = P_+, & [Q_-, Q_+] = Z_{-+}.
\end{align*}
\end{definition}

This definition has been chosen  to exploit light-cone coordinates and to easily identify  the action of Lorentz boosts on various expressions later in this paper.  In all, we see that the generators consist of a pair of right-handed and left-handed vectors, a  pair of right-handed and left-handed  Majorana--Weyl spinors and a Lorentz scalar.  It is a direct exercise to check that the appropriately graded Jacobi identity holds. For brevity, we will refer to the \emph{$\Z_2^2$-supertranslation algebra}.

\subsection{A $\Z_2^2$-graded Cartesian space realisation}
One could   use the  Baker--Campbell--Hausdorff formula to formally `integrate' the $\Z_2^2$-Lie algebra structure to obtain a $\Z_2^2$-Lie group, i.e., a group object in the category of $\Z_2^2$-manifolds (for group objects in a category see \cite{MacLane:1998}).  To do this carefully, just as in the standard supercase,  one requires the use of the functor of points or $\Lambda$-points (see \cite{Bruce:2020}). At the time of writing this paper,  $\Z_2^2$-Lie groups remain almost completely unexplored.  However, here we can make an educated postulate of the required $\Z_2^2$-space (see Example \ref{Exp:Z2nCartSp}). 
\begin{definition}\label{Def:Z22Mink}
The $\Z_2^2$-manifold that comes equipped with global coordinate systems of the form
 $$\{\underbrace{x^-}_{(0,0)}, ~ \underbrace{x^+}_{(0,0)},  ~ \underbrace{z}_{(1,1)},~\underbrace{\theta_-}_{(0,1)} ,~ \underbrace{\theta_+}_{(1,0)}  \},$$
will be referred to as \emph{$d=2$, $\mathcal{N} = (1,1)$ $\Z_2^2$-Minkowski spacetime} and shall be denoted as $\cM_2^{(1,1)}$. 
\end{definition}
Clearly, from the definition $\cM_2^{(1,1)} \cong \R^{2|1,1,1} = \big( \R^2, \, C^\infty_{\R^2}[[z, \theta_-, \theta_+]] \big)$.  Under Lorentz boosts these coordinate transform as 
\begin{align}\label{Eqn:Lorentz}\nonumber
& x^-   \mapsto \rme^{-\beta} \, x^-, 
& x^+   \mapsto \rme^{\beta} \, x^+,\\ \nonumber 
& z \mapsto z, &\\
& \theta_- \mapsto \rme^{-\half \beta}\, \theta_-,
& \theta_+ \mapsto \rme^{\half \beta}\, \theta_+,
\end{align}  
again, $\beta$ is the rapidity. Thus, the underlying reduced manifold is two-dimensional Minkowski spacetime, here explicitly given in light-cone coordinates.  For brevity, we will refer to \emph{$\Z_2^2$-Minkowski spacetime}.\par  
The following vector fields form a representation of the $\Z_2^2$-supertranslation algebra (see Definition \ref{Def:Z22SuperAlg})
\begin{align}\nonumber
& P_- = \frac{\partial}{\partial x^-}, && P_+ = \frac{\partial}{\partial x^+},\\
\nonumber 
& Z_{-+} = \frac{\partial}{\partial z}.&&\\
& Q_- = \frac{\partial }{\partial \theta_-} + \frac{1}{2} \theta_- \frac{\partial}{\partial x^-}  - \frac{1}{2} \theta_+ \frac{\partial}{\partial z},
&&Q_+ = \frac{\partial }{\partial \theta_+}  + \frac{1}{2} \theta_+ \frac{\partial}{\partial x^+} + \frac{1}{2} \theta_- \frac{\partial}{\partial z}.
\end{align}
The reader can easily verify that these vector fields have the right transformation properties under a Lorentz boost. \par
 The $\Z_2^2$-supertranslations acting on $\cM_2^{(1,1)}$ are parametrised by
 $$\{\underbrace{\lambda^-}_{(0,0)}, ~ \underbrace{\lambda^+}_{(0,0)}, ~  \underbrace{\mu}_{(1,1)}, ~ \underbrace{\epsilon_-}_{(0,1)}, ~ \underbrace{\epsilon_+}_{(1,0)}  \},$$
 and are explicitly given by
 \begin{align}\label{Eqn:SupTrans}\nonumber
 & x^- \mapsto x^-  + \lambda^- + \frac{1}{2}\, \epsilon_- \theta_- ,
 && x^+ \mapsto x^+ + \lambda^+ + \frac{1}{2}\,  \epsilon_+ \theta_+ ,\\ \nonumber
 & z \mapsto z' = z + \mu + \frac{1}{2}\big(\epsilon_+ \theta_-   -  \epsilon_- \theta_+\big),&&  \\ 
 & \theta_- \mapsto   \theta_- + \epsilon_-,
 && \theta_+ \mapsto   \theta_+ + \epsilon_+.
 \end{align}
\begin{remark}
The $\Z_2^2$-manifold constructed here should be compared with Tolstoy's $\Z_2^2$-graded spacetimes, see \cite{Tolstoy:2014a}. The key difference is that  our coordinates on the underlying Minkowski spacetime are of degree $(0,0)$ as where Tolstoy considers coordinates of degree $(1,1)$.
\end{remark}

\subsection{Integration on $\cM_2^{(1,1)}$ and the invariance of the coordinate volume}\label{SubSec:Int} Integration on general $\Z_2^n$-manifolds is still very much work in progress. However, the situation for low dimensional $\Z_2^2$-domains is understood (see \cite{Poncin:2016}). We will present the details of the soundness of the definition of the Berezin integral in Appendix \ref{App:BerInt}. Here we give the definition and prove that the coordinate Berezin volume is well-behaved with respect  to Lorentz boost and $\Z_2^2$-supertranslations. Two-dimensional Minkowski spacetime is orientable and we will assume without further reference, that an orientation has been chosen.  \par 
 An \emph{integrable Berezin section} on $\cM_2^{(1,1)}$ is a compactly supported Berezin section (not necessarily homogeneous)
$${\Bsigma}(x^-,x^+,z,\theta_-, \theta_+) = \mathrm{D}[x^-,x^+,z,\theta_-, \theta_+, ] \, \sigma(x^-,x^+,z,\theta_-, \theta_+),$$
such that there is no term of the form $z \, \sigma_z(x^-,x^+)$ in the function\footnote{or more generally the $\Z_2^2$-superfield, see Subsection \ref{SubSec:SupFie}.} $\sigma(x^-,x^+,z,\theta_-, \theta_+)$. Here $ D[x^-,x^+,z,\theta_-, \theta_+]$ is the  \emph{coordinate Berezin volume} and transforms under general coordinate changes as
$$ \mathrm{D}^\prime[x'^-,x'^+,z',\theta'_-, \theta'_+]= \mathrm{D}[x^-,x^+,z',\theta_-, \theta_+] \, \Z_2^2 \Ber\left( \frac{\partial(x'^-,x'^+, z', \theta'_-, \theta'_+)}{\partial(x^-,x^+, z,\theta_-, \theta_+)}\right),$$
where $\Z_2^2 \Ber$ is the $\Z_2^2$-graded generalisation of the Berezinian (see \cite{Covolo:2015,Covolo:2012,Poncin:2016} for further details). There is some slight abuse of notation here, as in general, we need not assume that the `primed' coordinates on the reduced manifold are light-cone coordinates. We will, by convention, take the coordinate Berezin volume to carry $\Z_2^2$-degree $(0,0)$. Other conventions are possible. \par 
The \emph{Berezin integral} of $\Bsigma$ is then defined as 
\begin{align}\label{Eqn:BerInt}\nonumber
\int_{\cM_2^{(1,1)}} \Bsigma &:= \int_{\cM_2^{(1,1)}}\mathrm{D}[x^-,x^+, z, \theta_-, \theta_+] \, \big( \sigma(x^-,x^+) + \theta_- \, \sigma_+(x^-,x^+) + \theta_+ \, \sigma_-(x^-,x^+) \\ \nonumber 
 &+ \theta_- \theta_+ \, \sigma_{+-}(x^-,x^+) + \cO(z)\big) \\
& = \int_{\R^2}\rmd x^+\, \rmd x^- \, \sigma_{+-}(x^-,x^+).
\end{align}
The final integral over $\R^2$ is well-defined as we have taken all the components of an integrable Berezin density to be compactly supported. In particular, $\sigma_{+-}$ is compactly supported. Note that to be consistent, we take the symbol $\int_{\cM_2^{(1,1)}}$ to carry $\Z_2^2$-degree $(1,1)$. \par 
\smallskip 

As we will not consider general coordinate transformations in this paper, we will not need to be more explicit here with the $\Z_2^2$-Berezinian and we simply point the interested reader to the literature and Appendix \ref{App:Ber}. However, we will show that the coordinate Berezin volume  is invariant under (infinitesimal)  $\Z_2^2$-supertranslations. To do this we will use the generalisation of Liouville's formula. 
\begin{proposition}\label{Prop:Invariance}
The coordinate  Berezin volume on $\cM_2^{(1,1)}$ is invariant under   $\Z_2^2$-supertranslations and Lorentz boosts. 
\end{proposition}
\begin{proof}
To show that $\mathrm{D}[x^-, x^+,z, \theta_-,\theta_+]$ is invariant we explicitly evaluate the $\Z_2^2$-Berezinian of the required Jacobi matrix
and show that it is equal to $1$. First, for the $\Z_2^2$-supertranslations \eqref{Eqn:SupTrans}
\smallskip 

\begin{align*}
J_{S} & = \renewcommand\arraystretch{1.5}  \begin{pmatrix}
    1 & 0  & 0 & -\half \epsilon_- &0   \\
    0 & 1  &  0 & 0& -\half \epsilon_+  \\
    0 & 0  &  1 & \half \epsilon_+ & -\half \epsilon_- \\
    0 & 0  &  0 & 1 & 0 \\
    0 & 0  &  0& 0 & 1 \\
\end{pmatrix}\\
\\
& = \renewcommand\arraystretch{1.5}  \begin{pmatrix}
    1 & 0  &  0 &0 & 0  \\
    0 & 1  & 0 &0 & 0 \\
    0 & 0  &  1 & 0 & 0 \\
    0 & 0  &  0 & 1 & 0 \\
    0 & 0  &  0 & 0 & 1 \\
\end{pmatrix} + \renewcommand\arraystretch{1.5}  \begin{pmatrix}
    0 & 0  & 0 & -\half \epsilon_- &0   \\
    0 & 0  &  0 & 0&0  \\
    0 & 0  &  0 & 0 & -\half \epsilon_- \\
    0 & 0  &  0 & 0 & 0 \\
    0 & 0  &  0& 0 & 0 \\
\end{pmatrix}
+ \renewcommand\arraystretch{1.5}  \begin{pmatrix}
   0 & 0  & 0 & 0 &0   \\
    0 & 0  &  0 & 0& -\half \epsilon_+  \\
    0 & 0  &  0 & \half \epsilon_+ & 0 \\
    0 & 0  &  0 & 0 & 0 \\
    0 & 0  &  0& 0 & 0 \\
\end{pmatrix}.
\end{align*}
\smallskip

Thus, we write $J_S = \Id_{5 \times 5} + A_- + A_+$, noting that the matrices $A_-$ and $A_+$ are infinitesimal. Thus, we can use the relation between the $\Z_2^2$-Berezinian and the $\Z_2^2$-trace (see \cite[Section 6]{Covolo:2012}) to deduce that
$$\Z_2^2\Ber(J_S) = 1 + \Z_2^2\tr(A_-) + \Z_2^2\tr(A_-) = 1,$$
as the $\Z_2^2$-graded trace is essentially the sum of $\pm$ the diagonal entries (see \cite[Section 2]{Covolo:2012}). Thus, the coordinate Berezin volume is invariant under $\Z_2^2$-supertranslations. As for Lorentz boosts \eqref{Eqn:Lorentz}, the  Jacobi matrix is 
\smallskip 

$$J_{L}  = \renewcommand\arraystretch{1.5}  \begin{pmatrix}
    \rme^{-\beta} & 0  & 0&0 & 0  \\
    0 & \rme^\beta  &  0 &0 & 0 \\
    0 & 0  & 1 & 0 & 0 \\
    0 & 0  &  0 & \rme^{-\half \beta} & 0 \\
    0 & 0  &  0 & 0& \rme^{\half \beta} \\
\end{pmatrix}.$$
\smallskip 

\noindent As the $\Z_2^2$-Berezinian of a diagonal matrix is just the product of the diagonal entries we see that $\Z_2^2\Ber(J_L) = 1$ (this follows from the definition, see \cite{Covolo:2012}).  Thus, the coordinate Berezin volume is invariant under Lorentz boosts. 
\end{proof}

\subsection{The Covariant Derivatives}
In standard supersymmetry, the origin of the SUSY covariant derivatives is the fact that the partial derivatives for the odd coordinates of superspace are not invariant under supertranslations. The same is true in the current situation and it is easy enough to see that the required \emph{$\Z_2^2$-SUSY covariant derivatives} on $\cM_2^{(1,1)}$ are
\begin{align}\label{Eqn:SUSYCovDer}
& D_- = \frac{\partial }{\partial \theta_-} - \frac{1}{2} \theta_- \frac{\partial}{\partial x^-}  +  \frac{1}{2}\theta_+  \frac{\partial}{\partial z},
&& D_+ = \frac{\partial }{\partial \theta_+} - \frac{1}{2} \theta_+ \frac{\partial}{\partial x^+} - \frac{1}{2} \theta_- \frac{\partial}{\partial z}.
\end{align}
Direct calculation gives the expected results
\begin{align*}
& [D_-, D_-] = - P_- ,\\
& [D_+, D_+] = -P_+,\\
& [D_-, D_+] = -Z_{-+},
\end{align*}
and, as a covariant derivative, we have 
\begin{align*}
&[D_-, Q_-] = [D_-, Q_+] =0,\\
&[D_+, Q_-] = [D_+, Q_+] =0.
\end{align*}

\subsection{Scalar $\Z_2^2$-superfields}\label{SubSec:SupFie}
By scalar, we will mean a Lorentz scalar with respect to $2$-dimensional Lorentz boosts of the source. A \emph{$\Z_2^2$-superfield}  is a map
$$\Phi : \cM_2^{(1,1)} \longrightarrow M,$$
where $M$ is a smooth manifold, or more generally a $\Z_2^2$-manifold, or possibly some other ``space'' for which we can make sense of maps.  Note that we take all maps and not just those that preserve the $\Z_2^2$-degree. To make this precise, we need the internal homs, however, we will work formally as is customary in physics. The interested reader can consult Lledo \cite{Lledo:2017} for details of the standard supercase. That is, we will allow ourselves the luxury of allowing external parameters that carry non-trivial $\Z_2^2$-degree.   Let us fix some finite-dimensional smooth manifold $M$, and consider an open $U \subset M$ ``small enough'' so that we can employ local coordinates $y^a$. Then we write (with the standard abuses of notation)
$$\Phi^*y^a  = \Phi^a(x^-,x^+,z, \theta_-, \theta_+).$$
We will avoid global issues in this paper and just remark that any map $\Phi$ has such local representative and that a family of such local maps can be glued together to form a global map. \par
We will take the $\Z_2^2$-superfield to be a scalar with respect to the $2$-d Lorentz transformations. Moreover, because we have taken the target to be a manifold, the  $\Z_2^2$-superfields we consider here are $\Z_2^2$-degree zero.  Then to first order in $z$,
\begin{align}\label{Eqn:SupField}\nonumber
 \Phi^a(x^-,x^+ ,z,\theta_- , \theta_+) &= X^a(x^-,x^+) + \theta_- \, \psi^a_+(x^-,x^+) + \theta_+ \, \psi^a_-(x^-,x^+) + \theta_- \theta_+ \ F^a_{+-}(x^-,x^+)\\ 
&+ z G^a(x^-,x^+) + z \theta_- \,\chi^a_+(x^-,x^+) + z \theta_+ \,\chi^a_-(x^-,x^+)+  z \theta_- \theta_+\,  Y^a_{+-}(x^-,x^+) \, + \mathcal{O}(z^2).
\end{align}
In general, we have a formal power series in $z$, that is, we have an infinite number of component fields. However, for our later purposes, it is sufficient to truncate to expressions being at most linear in $z$.  Then  first four types of fields in a multiplet are
$$\big( \underbrace{X^a}_{(0,0)},~  \underbrace{F^d_{+-}}_{(1,1)},~ \underbrace{\psi^b_+}_{(0,1)}, ~  \underbrace{\psi^c_-}_{(1,0)}\big). $$
For the $\Z_2^2$-superfield to be a Lorentz scalar we require that, under Lorentz boosts
\begin{align*}
& \psi^b_+ \mapsto \rme^{\half \beta}\,\psi^b_+ , && \psi^c_- \mapsto \rme^{-\half \beta}\,\psi^c_-,
\end{align*} 
while $X^a$ and $F^b_{+-}$ are Lorentz scalars. That is we have  bosons, exotic boson that anticommute with the fermions while commuting with the bosons, and  left-handed and right-handed  fermions that mutually commute. More correctly, one should speak of $\Z_2^2$-commutativity. \par
The next four components are  (ordered via $\Z_2^2$-degree)
$$\big( \underbrace{Y^e_{+-}}_{(0,0)},~   \underbrace{G^h}_{(1,1)},~ \underbrace{\chi^f_-}_{(0,1)}, ~  \underbrace{\chi^g_+}_{(1,0)}\big). $$
Again we have bosons, exotic bosons, and mutually commuting left-handed and right-handed  fermions.
\begin{remark}
It is also possible to consider $\Z_2^2$-superfields that transform non-trivially under Lorentz boosts. For example, spinor-valued maps can be made sense of. Moreover, $\Z_2^2$-superfields that carry a non-zero $\Z_2^2$-degree are also possible and will be considered later.
\end{remark}
We define the action of  $\Z_2^2$-supertranslations on a $\Z_2^2$-superfield via the Lie derivative. Specifically, the $\Z_2^2$-supersymmetry transformations are defined as $\delta \Phi^a :=\big( \epsilon_- Q_- + \epsilon_+ Q_+  \big) \Phi^a $.  Then, to lowest order in the field, the component form of $\Z_2^2$-supersymmetry is given by 
\begin{align}\label{Eqn:CompZ22SUSY}\nonumber
& \delta X^a \approx  \epsilon_- \psi^a_+ + \epsilon_+ \psi^a_-, &
 & \delta F^b_{+-} \approx - \half \big ( \epsilon_- \partial_- \psi^b_- + \epsilon_+ \partial_+ \psi^b_+ \big), \\
& \delta \psi^c_+ \approx - \epsilon _- \half \partial_- X^c + \epsilon_+ F^c_{+-}, && \delta \psi^d_- \approx - \epsilon_ + \half \partial_+ X^d + \epsilon_- F^d_{+-},
\end{align}
where we have used the shorthand $\partial_- = \frac{\partial}{\partial x^-}$ and $\partial_+ = \frac{\partial}{\partial x^+}$.
\begin{lemma}\label{Lem:DSupField}
Let $\Phi^a$ be a degree zero $\Z_2^2$-superfield (see \eqref{Eqn:SupField}), then 
\begin{align*}
D_-\Phi^a & =\psi^a_+ + \theta_+  F^a_{+-} - \frac{1}{2}\theta_- \partial_- X^a - \frac{1}{2}\theta_- \theta_+  \partial_- \psi^a_- + \frac{1}{2} \theta_+ G^a + \frac{1}{2}\theta_- \theta_+ \chi^a_+ + \cO(z), \\
D_+\Phi^b & =  \psi^b_- + \theta_- F_{+-}^b - \frac{1}{2}\theta_+ \partial_+ X^b - \frac{1}{2}\theta_- \theta_+ \partial_+\psi^b_+  - \frac{1}{2} \theta_+ G^a - \frac{1}{2}\theta_- \theta_+ \chi^b_- + \cO(z).
\end{align*}
\end{lemma}
\begin{proof}
This follows from direct computation using \eqref{Eqn:SupField} and \eqref{Eqn:SUSYCovDer}.\\
\end{proof}

\section{$\Z_2\times \Z_2$-graded sigma models}\label{Sec:SigmaModels}
\subsection{Sigma models with non-propagating exotic bosons}
We are now in a position to generalise  supersymmetric sigma models to our double-graded setting. The target space we will take as $n$-dimensional Minkowski spacetime $\mathcal{M}{ink}^n$.  A  $\Z_2^2$-superfield is then a map $\Phi :  \cM_2^{(1,1)} \rightarrow \mathcal{M}{ink}^n$, and using coordinates $y^a$ on the target we  write $\Phi^*y^a := \Phi^a(x^-,x^+, z,\theta_-, \theta_+)$.   For any model to be well-defined, the Lagrangian Berezin section $\cL[\Phi] := \mathrm{D}[x^-, x^+, z, \theta_-  \theta_+] \, L(\Phi)$ needs to be an integrable Berezin section (see Subsection \ref{SubSec:Int}).  This means that the $\Z_2^2$-superfields in questions must be compactly supported, that is each of their components is compactly supported. Furthermore, the Lagrangian Berezin section must not contain a term of the form $z\,L(x^-x^+)$. This condition is independent of the coordinates chosen (see Appendix \ref{App:BerInt}). In light of Proposition \ref{Prop:Invariance} and the form of the $\Z_2^2$-supertranslations, this condition is also $\Z_2^2$-supersymmetric.  Depending on the exact form of the Lagrangian Berezin section, being integrable will place constraints on the $\Z_2^2$-superfields involved.  \par 
The most general action involving a scalar $\Z_2^2$-superfield is
$$S[\Phi] := \int_{\cM_2^{(1,1)}} \mathrm{D}[x^-, x^+, z, \theta_-, \theta_+]\, \mathcal{K}(\Phi, D_-\Phi, D_+ \Phi),$$
where $\mathcal{K}(\Phi, D_-\Phi, D_+ \Phi)$ has to be a Lorentz scalar and of degree $(1,1)$.  We specialise to specific actions. The simplest Lagrangian Berezin section that is second-order in the covariant derivatives, with a degree $(1,1)$ Lagrangian is 
$$\mathrm{D}[x^-, x^+, z, \theta_-, \theta_+]\, D_-\Phi^a D_+\Phi^b \eta_{ba}.$$
Being integrable means that the $\Z_2^2$-superfield (see \eqref{Eqn:SupField}) cannot contain the terms $z \theta_ -\chi^a_+$ and $z \theta_ +\chi^a_-$. Later we wish to add interactions and so we will insist that the $\Z_2^2$-superfield does not contain the term $zG^a$ either.  For example, the quadratic potential  $\half \Phi^a \Phi^b \eta_{ba}$ would contain the term $z X^a G^b\eta_{ba}$ and so cannot be integrable. 
\begin{definition}
Let $M$ be a smooth manifold and let $\Phi: \cM_2^{(1,1)} \rightarrow M$  be a (scalar) $\Z_2^2$-superfield.  Then $\Phi$ is said to be \emph{$z$-constrained} if  and only if its representative $\Phi^a(x^-, x^+, z,\theta_- , \theta_+)$  (see \eqref{Eqn:SupField}) does not contain the terms $z \theta_-\chi^a_+(x^-, x^+)$, $z \theta_+\chi^a_-(x^-, x^+)$, and $z G^a(x^-, x^+)$.
\end{definition}
\begin{remark}
A small modification of  the proof of Proposition \ref{Prop:NoZTermSigma} shows that being $z$-constrained is independent of the choice of coordinates on the source. However, we will \emph{not} concern ourselves here with general coordinate transformations on the source.  
\end{remark}
\begin{definition}\label{Def:SigmaModel} Let $\mathcal{M}{ink}^n$ be  $n$-dimensional Minkowski spacetime, and let $\Phi \in \InHom_{\,c}(\cM_2^{(1,1)} , \mathcal{M}{ink}^n )$ be a $z$-constrained $\Z_2^2$-superfield with compact support.  Then the \emph{$\Z_2^2$-graded linear sigma model action} is
$$\mathrm{S}_0[\Phi] := \int_{\cM_2^{(1,1)}} \mathrm{D}[x^-,x^+,z, \theta_-,\theta_+] \, D_-\Phi^a D_+\Phi^b \eta_{ba}.$$
\end{definition}
\begin{remark}
For simplicity, we have not included a potential term and so exclude Landau--Ginsburg type models from our discussion for the moment.
\end{remark}
\begin{proposition}\label{Prop:WellDefLinSig}
The $\Z_2^2$-graded linear sigma model action is well-defined.
\end{proposition}
\begin{proof}
This is clear from the preceding discussion, principally that the Lagrangian Berezin section is integrable, i.e., there is no $z L(x^-, x^+)$ term in the integrand and all the fields involved have compact support.  \\
\end{proof}
\begin{theorem}\label{Thm:SigLinInv}
The $\Z_2^2$-graded linear sigma model action is invariant under $\Z_2^2$-supertranslations (see \eqref{Eqn:SupTrans}). 
\end{theorem}
\begin{proof}
The Lagrangian  is built from $D_-\Phi^a$ and $D_+\Phi^b$, both of which by construction are invariant under  $\Z_2^2$-supertranslations.  We have already shown that the coordinate volume is invariant, Proposition  \ref{Prop:Invariance}. Thus, the action is invariant.
\end{proof}
We now proceed to present  the component form of this action. We only need to keep track of the $\theta_- \theta_+$ terms due to how the Berezin integral is defined. Thus, a quick calculation shows that
\begin{align*}
D_-\Phi^aD_+\Phi^b\eta_{ba} & \approx  \theta_- \theta_+ \left( \frac{1}{4}\big(\partial_ - X^a \partial_+ X^b\big)\eta_{ba}   +  \frac{1}{2}\big(\psi^a_+ \partial_+\psi^b_+ \big)\eta_{ba}   +  \frac{1}{2}\big(\psi^a_- \partial_-\psi^b_- \big)\eta_{ba}     - F^a_{+-}F^b_{+-} \eta_{ba}  \right).
\end{align*}
Hence the component form of this action is 
$$S_0 = \int_{\R^2}\rmd x^- \rmd x^+ \left( \frac{1}{4}\big(\partial_ - X^a \partial_+ X^b\big)\eta_{ba}   +  \frac{1}{2}\big(\psi^a_+ \partial_+\psi^b_+ \big)\eta_{ba}   +  \frac{1}{2}\big(\psi^a_- \partial_-\psi^b_- \big)\eta_{ba}     - F^a_{+-}F^b_{+-} \eta_{ba}\right),$$
which is a massless Wess--Zumino type model (see \cite{Wess:1974}) in $2$-d with an axillary field (in light-cone coordinates). This action is invariant under the $\Z_2^2$-supersymmetries \eqref{Eqn:CompZ22SUSY}.  Note that as we have insisted on using $z$-constrained $\Z_2^2$-superfields that $\Z_2^2$-supersymmetry does not mix component fields that are of order-zero in $z$ with component fields of higher-order in $z$.  Thus, the field content in this model is minimal.  A direct calculation  gives the expected equations of motion for the component fields 
\begin{align}\label{Eqn:EQMFree}\nonumber 
& \partial_- \partial_+ X^a =0, && F^b_{-+} =0,\\
&\partial _+\psi^c_+ =0, && \partial _-\psi^d_- =0.
\end{align} 
Alternatively, we can work directly with the $\Z_2^2$-space Euler--Lagrange equations\footnote{At this stage we postulate rather than derive these equations. Details about the calculus of variations in this setting remain to be uncovered. }. For notational ease we define $\Phi^a_{\mp} = D_{\mp}\Phi^a$
$$D_-\left( \frac{\partial L}{\partial \Phi_-^a}\right) +D_+\left( \frac{\partial L}{\partial\Phi_+^a}\right)= D_- D_+ \Phi^a  +D_+ D_- \Phi^a = 0.$$
Expanding this out to lowest-order (remembering we have $z$-constrained $\Z^2_2$-superfields) yields  
$$2 F^a_{-+} - \theta_+ \partial_+\psi^a_+ - \theta_- \partial_-\psi^a_- + \frac{1}{2}\theta_- \theta_+ \, \partial_+ \partial_- X^a =0,$$
and so we recover \eqref{Eqn:EQMFree}. \par 
We now extend the kind of sigma model we are studying by now allowing the target manifold to be a finite-dimensional (pseudo-)Riemannian manifold $(M,g)$. This leads to the following definition. 
\begin{definition}\label{Def:NonLinSigmaModel} Let $(M,g)$ be  an $n$-dimensional (pseudo-)Riemannian manifold, and let $\Phi \in \InHom_{\,c}(\cM_2^{(1,1)} , M )$ be a $z$-constrained $\Z_2^2$-superfield with compact support.  Then the \emph{$\Z_2^2$-graded non-linear sigma model action} is
$$\mathrm{S}_0[\Phi] := \int_{\cM_2^{(1,1)}} \mathrm{D}[x^-,x^+,z, \theta_-,\theta_+] \, D_-\Phi^a D_+\Phi^b g_{ba}(\Phi).$$
\end{definition}
\begin{theorem}
The $\Z_2^2$-graded non-linear sigma model action is well-defined and is invariant under\newline  $\Z_2^2$-supertranslations.
\end{theorem}
\begin{proof}We break the proof up into the two statements.\\
\noindent\emph{Well-defined:}
Taylor expanding the metric we see that $g_{ba}(\Phi)$,
$$g_{ba}(\Phi) = g_{ba}(X) + \theta_ + \psi^c_- \frac{\partial g_{ba}(X)}{\partial X^c} +\theta_ - \psi^c_+ \frac{\partial g_{ba}(X)}{\partial X^c} + \theta_-\theta_ + F^c_{+-} \frac{\partial g_{ba}(X)}{\partial X^c} + \theta_-\theta_+ \psi^c_+\psi^d_- \frac{\partial^2 g_{ba}(X)}{\partial X^d \partial X^c} + \cdots, $$
we observe that this is itself $z$-constrained as we have taken $\Phi$ to be $z$-constrained. Thus, the Lagrangian $L =  D_-\Phi^a D_+\Phi^b g_{ba}(\Phi)$ is $z$-constrained (see Proposition \ref{Prop:WellDefLinSig}). Moreover, because  products of compactly supported $\Z_2^2$-superfields  and their derivatives are again compactly supported,  this Lagrangian is compactly supported.

\smallskip
\noindent\emph{Invariance under $\Z_2^2$-supertranslations:} This is evident as the canonical volume is invariant (Proposition \ref{Prop:Invariance}) and the Lagrangian is built from invariant objects; $\Phi^a$, $D_-\Phi^b$ and $D_+\Phi^c$.\\
\end{proof}

\bigskip

Now we address the possibility of including a superpotential.  The immediate problem is that a smooth function $U(\Phi)$ of a degree zero $\Z_2^2$-superpotential is itself degree zero. Thus, it cannot be used as a potential in the linear or the non-linear $\Z_2^2$-graded sigma models. A similar situation occurs in $N= 1$ supersymmetric mechanics where one cannot directly include a potential term, i.e., a term like $U(x)\psi$ is odd and cannot  be included in the action. Manton's solution (see \cite{Manton:1999}) was to include an odd constant in the potential, or in other words, to consider a Grassmann odd potential of the form $W(x) = \alpha \,U(x)$, where $\alpha$ is the odd constant, i.e., an odd element of a chosen Grassmann algebra. This kind of complication was also noticed in $\Z_2^2$-mechanics by Aizawa et al. \cite{Aizawa:2020c}.  We propose a similar solution, by taking a degree $(1,1)$ $\Z_2^2$-superpotential $W(\Phi)$. Mathematically we can understand this as a composition of internal homs. More informally, we are allowing graded constants in both the definition of $\Phi$ and $W(-)$.  The interaction term is then
$$S_{int}[\Phi] = - \int_{\cM_2^{(1,1)}} \mathrm{D}[x^-, x^+, z,\theta_-, \theta_+] \,  W(\Phi),$$
which, assuming $\Phi$ is $z$-constrained, is well-defined and  is invariant under both Lorentz transformations and $\Z_2^2$-supertranslations. Then the total action, taking the target to be Minkowski spacetime for simplicity,  written in component form is
\begin{align*}
 S_0[\Phi] + S_{int}[\Phi]  & = \int_{\R^2}\rmd x^- \rmd x^+ \left( \frac{1}{4}\big(\partial_ - X^a \partial_+ X^b\big)\eta_{ba}   +  \frac{1}{2}\big(\psi^a_+ \partial_+\psi^b_+ \big)\eta_{ba}   +  \frac{1}{2}\big(\psi^a_- \partial_-\psi^b_- \big)\eta_{ba}  \right.\\ 
& \left.    - F^a_{+-}F^b_{+-} \eta_{ba} - F^a_{+-}\frac{\partial W(X)}{\partial X^a} - \psi^a_- \psi^b_+ \frac{\partial^2 W(X) }{\partial X^b \partial X^a}\right).
\end{align*}
The equation of motion for the exotic boson is $F^b_{+-}\eta_{ba} = - \half \frac{\partial W(X)}{\partial X^a}$. This can then be used to eliminate the exotic boson and the action can be written as
\begin{align*}
 S[\Phi]   & = \int_{\R^2}\rmd x^- \rmd x^+ \left( \frac{1}{4}\big(\partial_ - X^a \partial_+ X^b\big)\eta_{ba}   +  \frac{1}{2}\big(\psi^a_+ \partial_+\psi^b_+ \big)\eta_{ba}   +  \frac{1}{2}\big(\psi^a_- \partial_-\psi^b_- \big)\eta_{ba}  \right.\\ 
& \left.   + \frac{1}{4} \frac{\partial W(X)}{\partial X^b}\frac{\partial W(X)}{\partial X^a}\eta^{ab}  -  \psi^a_- \psi^b_+  \,\frac{\partial^2 W(X)}{\partial X^b\partial X^a}\right).
\end{align*}
This is very similar to the standard case of supersymmetry. However, we replace the exotic boson with an exotic potential, the physical meaning of which is somewhat obscured. As an example, one could pick a potential such that $W(\Phi)= \alpha U(\Phi)$, where $\alpha$ is a degree $(1,1)$ constant (that transforms as a Lorentz scalar).  Then, the previous action becomes  
\begin{align*}
 S[\Phi]   & = \int_{\R^2}\rmd x^- \rmd x^+ \left( \frac{1}{4}\big(\partial_ - X^a \partial_+ X^b\big)\eta_{ba}   +  \frac{1}{2}\big(\psi^a_+ \partial_+\psi^b_+ \big)\eta_{ba}   +  \frac{1}{2}\big(\psi^a_- \partial_-\psi^b_- \big)\eta_{ba}  \right.\\ 
&  \left.   + \frac{\alpha^2}{4} \frac{\partial U(X)}{\partial X^b}\frac{\partial U(X)}{\partial X^a}\eta^{ab} -  \alpha \, \psi^a_- \psi^b_+  \,\frac{\partial^2 U(X)}{\partial X^b\partial X^a}\right).
\end{align*}
To avoid complications with the equations of motion for the bosonic sector, it makes sense to assume that $\alpha^2 -1 =0$, so $\alpha$ belongs to a one-dimensional Clifford algebra. Notice that the parameter $\alpha$ appears in the (on-shell) component $\Z_2^2$-supersymmetry transformations. We remark that Clifford algebra-valued parameters have already appeared in the context of supersymmetry (see \cite{Kuzentosova:2008}).  Moreover, Akulov \& Duplij, in the setting of supersymmetric mechanics, proposed that classical solutions to the equations of motion should also depend on Grassmann algebra-valued constants (see \cite{Akulov:1999}). \par 
Rather than explore the properties of general models we will study a specific  example.  In particular, we take the target manifold to be $\R$ and consider the potential
$$W(\Phi) : =  2 \alpha \big(1- \cos( \Phi/2)  \big).$$
This will allow us to highlight some features of these models without the clutter of various indices. The action we will consider  is
\begin{equation}\label{Egn:Z22SinGordAct}
\mathrm{S}[\Phi] := \int_{\cM_2^{(1,1)}} \mathrm{D}[x^-,x^+,z, \theta_-,\theta_+] \,\left( D_-\Phi D_+\Phi  - 2 \alpha \big(1- \cos( \Phi/2)  \big)~\right).
\end{equation}
The action \eqref{Egn:Z22SinGordAct} could be referred to as the \emph{$\Z_2^2$-graded sine-Gordon action}.  Note that the parameter $\alpha$ is required in the Yukawa-like coupling just on degree grounds. The component action (eliminating the auxiliary field via its equation of motion $2 F _{+-}= - \alpha \sin (X/2)$~) is 
\begin{equation}\label{EqnZ22SinGordActComp}
 S[\Phi]    = \int_{\R^2}\rmd x^- \rmd x^+ \left( \frac{1}{4}\partial_ - X \partial_+ X   +  \frac{1}{2}\psi_+ \partial_+\psi_+    +  \frac{1}{2}\psi_- \partial_-\psi_-  
  + \frac{1}{4} \sin^2(X/2) - \frac{1}{2} \alpha \, \psi_- \psi_+  \,\cos(X/2) \right).
  \end{equation}
The equations of motion for the component fields  are 
\begin{align}\label{Eqn:Z22SinGordEqCom}
 \nonumber &\partial_- \partial_+ X - \frac{1}{4}  \sin(X) - \frac{\alpha}{2} \,\psi_- \psi_+ \, \sin(X/2)=0,&&\\
& \partial_+ \psi_+  + \frac{\alpha}{2} \, \psi_- \cos(X/2)=0,
&& \partial_- \psi_- +  \frac{\alpha}{2} \, \psi_+ \cos(X/2)=0.
 \end{align}
 Note that setting $\psi_+ = \psi_- =0$ reduces the system to the classical sine-Gordon equation. These equations of motion can also be derived using the 
 $\Z_2^2$-space Euler--Lagrange equations, which now take the form
$$D_-\left( \frac{\partial L}{\partial \Phi_-}\right) +D_+\left( \frac{\partial L}{\partial\Phi_+}\right) - \frac{\partial W(\Phi)}{\partial \Phi} = 0.$$
For the $\Z_2^2$-graded sine-Gordon  action \eqref{Egn:Z22SinGordAct}, we have what could be called the \emph{$\Z_2^2$-graded sine-Gordon equation}
 $$D_- D_+ \Phi + D_+ D_- \Phi  + \alpha \, \sin(\Phi/2)=0.$$
 Expanding the above to lowest-order yields
 \begin{align*}
 &  \partial_- \partial_+ X + \alpha \, F_{+-} \cos(X/2) - \frac{\alpha}{2} \psi_- \psi_+  \sin(X/2)=0, && 2 F_{+-}+ \alpha \sin(X/2)=0,\\
 & \partial_+ \psi_+  + \frac{\alpha}{2} \, \psi_- \cos(X/2)=0,
&& \partial_- \psi_- +  \frac{\alpha}{2} \, \psi_+ \cos(X/2)=0.
 \end{align*}
Clearly this system reduces to \eqref{Eqn:Z22SinGordEqCom} when eliminating the auxiliary field from the equations of motion for the boson. The on-shell $\Z_2^2$-supersymmtery transformations take the form 
 \begin{align*}
 & \delta X  = \epsilon_- \psi_+ +   \epsilon_+ \psi_-,& \\
 & \delta \psi_+ = - \frac{1}{2}\left(\epsilon_- \,\partial_-X + \epsilon_+ \alpha \sin(X/2)  \right), &  \delta \psi_- = - \frac{1}{2}\left(\epsilon_+ \,\partial_+X + \epsilon_- \alpha \sin(X/2)  \right).
 \end{align*}
Note the aforementioned explicit presence of the parameter $\alpha$ in these transformations. We proceed to check that the component action is invariant under the $\Z_2^2$-supersymmetry transformations. As standard, the action itself is not invariant but only quasi-invariant, i.e., $\delta L = \partial_-V^- + \partial_+V^+ $.  Writing $L = L_0 + L_+ + L_- + L_1 + L_{-+}$ for each term in  the Lagrangian of \eqref{EqnZ22SinGordActComp}, a direct computation produces
\begin{align*}
& \delta L_0 = \frac{1}{4} \epsilon_-\big( \partial_- \psi_+ \partial_+ X + \partial_+ \psi_+ \partial_-X\big) + \frac{1}{4} \epsilon_+ \big( \partial_- \psi_- \partial_+ X + \partial_+ \psi_- \partial_-X\big),\\
& \delta L_+ = - \frac{1}{4}\epsilon_- \big(\partial_+ \psi_+ \partial_-X   - \psi_+ \partial _+ \partial _- X \big) - \frac{1}{4}\epsilon_+ \big( \alpha \sin(X/2)\partial_+ \psi_+  - \alpha \partial _+X  ~ \half \cos(X/2) \psi_+\big),\\
& \delta L_- = - \frac{1}{4}\epsilon_+ \big(\partial_- \psi_- \partial_+X   - \psi_- \partial _- \partial _+ X \big) - \frac{1}{4}\epsilon_- \big( \alpha \sin(X/2)\partial_- \psi_-  - \alpha \partial _-X  ~ \half \cos(X/2) \psi_-\big),\\
&\delta L_1 = \frac{1}{8}\epsilon_- \psi_+ \sin(X) + \frac{1}{8}\epsilon_+ \psi_-\sin(X),\\
&\delta L_{-+} = -\frac{1}{4}\epsilon_+ \big(\alpha \partial_+X ~ \cos(X/2)\psi_+ + \psi_-  \, \half \sin(X) \big) - \frac{1}{4} \epsilon_- \big(\psi_+  \, \half \sin(X) + \alpha  \partial_- X ~\cos(X/2) \psi_-\big).
\end{align*}
Collecting terms we see that
\begin{equation}\label{Eqn:DelLagSinGor}
\delta L =\partial_-V^- + \partial_+V^+ =  \partial_- \left(\frac{\epsilon_-}{4}\big(\psi_+ \partial_+X - \alpha \, \sin(X /2) \psi_- \big) \right) + \partial_+ \left(\frac{\epsilon_+}{4}\big(\psi_- \partial_-X - \alpha \, \sin(X /2) \psi_+ \big) \right).
\end{equation}
Thus, as expected, the component Lagrangian of the $\Z_2^2$-graded sine-Gordon action is quasi-invariant under the $\Z_2^2$-supersymmetry transformations.   Application of Noether's theorem yields the conserved currents. We write $\phi^A$ for the collection of (relevant) component fields and set $\phi^A_{\mp} = \partial_{\mp} \phi^A$.  Then a quick calculation produces
\begin{align*}
& J^{--} := \frac{\delta \phi^A}{\delta \epsilon_-}\left( \frac{\partial L}{\partial \phi^A_-}\right) - \frac{\partial V^-}{\partial \epsilon_-} = \half \alpha  \, \sin(X/2) \psi_- = \left( \frac{\delta \psi_-}{\delta \epsilon_-}\right)\psi_-,\\
& J^{-+} := \frac{\delta \phi^A}{\delta \epsilon_-}\left( \frac{\partial L}{\partial \phi^A_+}\right) - \frac{\partial V^+}{\partial \epsilon_-} = \half \partial_-X \, \psi_+ = {-}\left( \frac{\delta \psi_+}{\delta \epsilon_-}\right)\psi_+,\\
& J^{+-} := \frac{\delta \phi^A}{\delta \epsilon_+}\left( \frac{\partial L}{\partial \phi^A_-}\right) - \frac{\partial V^-}{\partial \epsilon_+} = \half \partial_+X \, \psi_- = {-}\left( \frac{\delta \psi_-}{\delta \epsilon_+}\right)\psi_-,\\
& J^{++} := \frac{\delta \phi^A}{\delta \epsilon_+}\left( \frac{\partial L}{\partial \phi^A_+}\right) - \frac{\partial V^+}{\partial \epsilon_+} = \half \alpha  \, \sin(X/2) \psi_+ = \left( \frac{\delta \psi_+}{\delta \epsilon_+}\right)\psi_+.
\end{align*}
Note that $J^{--}$ and $J^{-+}$ are of $\Z_2^2$-degree $(0,1)$, while $J^{+-}$ and $J^{++}$ are of $\Z_2^2$-degree $(1,0)$. Moreover, as expected, these currents are spinor in nature, i.e., under Lorentz boots they transform as
\begin{align*}
& J^{--} \mapsto \rme^{-\half \beta}\, J^{--}, && J^{-+} \mapsto \rme^{\frac{3}{2} \beta}\, J^{-+},\\
& J^{+-} \mapsto \rme^{-\frac{3}{2} \beta}\, J^{+-}, && J^{++} \mapsto \rme^{\half \beta}\, J^{++}.
\end{align*}
A quick calculation verifies  these currents are indeed conserved when the equations of motion hold, i.e.,
\begin{align*}
& \partial_- J^{--} + \partial_+ J^{-+} = \half \left(\partial_+ \psi_+ + \half\alpha \, \psi_-\cos(X/2)  \right)\partial_- X + \half \left( \partial_- \partial_+ X - \frac{1}{4}  \sin(X) - \frac{\alpha}{2} \,\psi_- \psi_+ \, \sin(X/2)\right)\psi_+,\\
& \partial_- J^{+-} + \partial_+ J^{++} = \half \left(\partial_- \psi_- + \half\alpha \, \psi_+\cos(X/2)  \right)\partial_+ X + \half \left( \partial_- \partial_+ X - \frac{1}{4}  \sin(X) - \frac{\alpha}{2} \,\psi_- \psi_+ \, \sin(X/2)\right)\psi_-,
\end{align*}
which clearly vanish on-shell (see \eqref{Eqn:Z22SinGordEqCom}). Given the importance of sine-Gordon and supersymmetric sine-Gordon in the theory of integrable systems  and soliton theory, we firmly believe this $\Z_2^2$-graded version deserves further study. It would certainly   be interesting to construct  solutions and compare them with solutions of the $N=1$ and $N=2$ supersymmetric sine-Gordon equations. It is clear that the fact that $\psi_+$ and $\psi_-$ are of different $\Z_2^2$-degree and the presence of a degree $(1,1)$ constant complicates finding classical solutions.   A detailed study of solutions is outside the scope of this paper. 

\subsection{A Model with propagating exotic bosons}
We now briefly examine a simple model in which the exotic boson is propagating. To do this we consider a scalar $\Z_2^2$-superfield that carries $\Z_2^2$-degree $(1,1)$. That is, we consider the target space to be $\R^{0|1,0,0}$.  Naturally, we insist that this $\Z_2^2$-superfield $\Psi \in \InHom_{\, c}(\cM_2^{(1,1)},\R^{0|1,0,0} )$  be compactly supported and $z$-constrained as before.
$$\Psi(x^-, x^+, z, \theta_-, \theta_+) = Y(x^-, x^+) + \theta_- \chi_+(x^-, x^+) + \theta_+ \chi_-(x^-, x^+) + \theta_- \theta_+ G(x^-, x^+) + \cdots,$$
where the degrees are 
$$\big( \underbrace{G}_{(0,0)}, ~   \underbrace{Y}_{(1,1)},~  \underbrace{\chi_-}_{(0,1)}, ~  \underbrace{\chi_+}_{(1,0)}\big). $$
We can directly include a potential by considering a formal power series in the $\Z_2^2$-superfield. We require that the potential be degree $(1,1)$ and so, assuming no constants of non-zero degree, the power series must  be of the form 
$$U(\Psi) := \sum_{k=0}^{n }\frac{1}{(2k+1)!} \,a_{2k+1} (\Psi)^{2k+1}, $$
where $a_{k+1}\in \R$ and $n \in \mathbb{N}$.
\begin{definition}\label{Def:ExoticScalar} Let $\Psi \in \InHom_{\,c}(\cM_2^{(1,1)} , \R^{0|1,0,0} )$ be a $z$-constrained $\Z_2^2$-superfield with compact support.  Then the \emph{$\Z_2^2$-graded  exotic scalar field theory action} is
$$\mathrm{S}[\Psi] := \int_{\cM_2^{(1,1)}} \mathrm{D}[x^-,x^+,z, \theta_-,\theta_+] \, \left( D_-\Psi D_+\Psi - \sum_{k=0}^n\frac{1}{(2k+1)!} \,a_{2k+1} (\Psi)^{2k+1} \right).$$
\end{definition}
The following theorem is evident given our earlier discussions.
\begin{theorem}
The $\Z_2^2$-graded  exotic scalar field theory action is well-defined, invariant under two-dimensional Lorentz boosts and $\Z_2^2$-supertranslations.
\end{theorem}
A direct calculation shows that the component of the action, ignoring the potential for simplicity, is
\begin{equation}\label{Eqn:FreeExoticAct}
S_0[\Psi] = \int_{\R^2} \rmd x^- \rmd x^+ \left( - \frac{1}{4} \partial_- Y \partial_+Y +\frac{1}{2}\chi_- \partial_ -\chi_-  + \frac{1}{2}\chi_+ \partial_+ \chi_+ + G^2 \right).
\end{equation}
A direct calculations shows that the equations of motion are
\begin{align}\label{Eqn:ExoEqm}\nonumber
& G=0, & 
 \partial_ -\partial_ + Y =0& \\ 
& \partial_- \chi_- = 0,
& \partial_+ \chi_+ = 0.               
\end{align}
 As before, these equations of motion can be deduced from the $\Z_2^2$-space Euler--Lagrange equations. The reader can quickly verify that the $\Z_2^2$-space equations of motion reduce to $D_- \Psi D_+\Psi= 0$ and that this agrees with the component equations of motion \eqref{Eqn:ExoEqm}. \par 
The component form of $\Z_2^2$-supersymmetry are 
\begin{align}\nonumber
& \delta G = - \frac{1}{2} \epsilon_- \partial_- \chi_- - \frac{1}{2} \epsilon_+ \partial_+ \chi_+, & 
 \delta Y = \epsilon_- \chi_+  + \epsilon_+ \chi_-, & \\ 
& \delta \chi_- = - \frac{1}{2} \epsilon_+ \partial_+ Y + \epsilon_- G,
& \delta \chi_+ = - \frac{1}{2} \epsilon_-\partial_-Y + \epsilon_+ G.               
\end{align}
A direct calculation (and using the equations of motion \eqref{Eqn:ExoEqm}) shows that
$$\delta L  = \partial_- V^- + \partial_+ V^+ = \partial_- \left( \frac{\epsilon_-}{4} \, \partial _+ Y \, \chi_+  \right) + \partial_+ \left( \frac{\epsilon_+}{4}\, \partial _- Y \, \chi_- \right),$$
and so, as expected, the component Lagrangian associated with the exotic action \eqref{Eqn:FreeExoticAct} is quasi-invariant under  $\Z_2^2$-supersymmetry. The two non-zero Noether currents are 
\begin{align*}
& J^{-+} = \frac{1}{2}\partial_-Y \chi_+, &&  J^{++} = \frac{1}{2}\partial_+Y \chi_-,
\end{align*} 
which are of $\Z_2^2$-degree $(0,1)$ and $(1,0)$, respectively. Under Lorentz boots they transform as spinors, i.e., 
\begin{align*}
& J^{-+} \mapsto \rme^{\frac{3}{2} \beta}\, J^{-+},&&  J^{+-} \mapsto \rme^{-\frac{3}{2} \beta}\, J^{+-}.
\end{align*}
It is easy to observe that these currents are indeed conserved on-shell, i.e., $\partial_+ J^{-+} =0$ and $\partial_- J^{+-} =0$.

\section{Concluding Remarks}\label{Sec:ConRem}
In this paper, we proposed a double-graded version of the $\mathcal{N} = (1,1)$ supertranslation algebra in two dimensions, built a $\Z_2^2$-space realisation thereof and used this to construct $\Z_2^2$-supersymmetric sigma models with either auxiliary or propagating exotic bosons. The focus has been on classical aspects of the theory, showing both the novel features and potential problems.  There is much that remains to be done, including quantisation, the construction sigma models with  general $\Z_2^2$-manifold target spaces (including $\Z_2^2$-Lie groups), models with higher dimensional sources, and the development of $\Z_2^n$-graded models ($n >2$). The immediate issue here is that the theory of integration on $\Z_2^n$-manifolds is not yet properly developed. Thus, at the time of writing, it is impossible to mimic superspace methods in any generality. Many aspects of the theory of $\Z_2^n$-manifolds are currently being developed.  None-the-less, in this paper, we have shown that relatively simple, yet novel  double graded classical models in two-dimensions exist.  An interesting feature is that the grading puts restrictions on the actions, quite independently of the symmetry. In particular, the Lagrangian Berezin sections must be degree $(1,1)$, so that the component action is degree $(0,0)$. This adds another layer of complication when directly trying to mimic superspace constructions. For example, to include interactions in some of the models, we require graded constants in the potentials.     \par
The physical relevancy of $\Z_2^2$-graded theories, and in particular $\Z_2^2$-supersymmetry is not at all clear. The exact r\^{o}le of exotic degree $(1,1)$ bosons in physics  is an open question. Potentially less confusing are the two sectors of mutually commuting spinors. When one has two independent spinors it is usually assumed that they mutually anticommute. However, this is an assumption that can be  dropped, and taking the spinors to mutually commute is consistent (see  Aste \& Chung \cite{Aste:2016}).  What would be a clear phenomenological signal that nature utilises $\Z_2^2$-gradings and this novel kind of supersymmetry? This is the key question and one that may be answered as further models are developed.

\section*{Acknowledgements}
The author thanks Norbert Poncin \& Eduardo Ibarguengoytia for fruitful discussions about integration on $\Z_2^n$-manifolds and related topics.  A special thank you goes to Steven Duplij for his encouragement and being the impetus for this work.

%%%%%%%%%%%%%%%%%%%%%%%%%%%%%%%%%%%%%%%%%%%%

\appendix

\section{The $\Z_2^2$-Berezinian}\label{App:Ber}
The $\Z_2^2$-Berezinian  is directly related to the theory of quasi-determinants as given by Gel'fand \& Retakh \cite{Gelfand:1991} and was first described in \cite{Covolo:2012}. The classical Berezinian is defined in terms of the determinant of its blocks (for details, see for example \cite{Manin;1997}). Similarly,  the $\Z_2^n$-Berezinian is built from the  $\Z_2^n$-graded determinant. However, for $n=1$ and $n =2$,  the $\Z_2^n$-graded determinant is the same as the classical determinant. For $n>2$ the situation is more complicated and the reader should consult \cite{Covolo:2012}.
\begin{definition}\label{Def:Z22Ber}
Let $X$ be a degree zero invertible matrix with entries from a $\Z_2^2$-graded, $\Z_2^2$-commutative, unital algebra $\mathcal{A}$, written in block form
\[
X=
\left(
\begin{array}{c|c}
A & B \\
\hline
C & D
\end{array}
\right)
.\]
Then the \emph{$\Z_2^2$-Berezinian} is the group homomorphism 
$$\Z_2^2\textrm{Ber} : \textrm{GL}^0(\mathcal{A}) \rightarrow (\mathcal{A}^0)^\times,$$
defined as
$$\Z_2^2\textrm{Ber}(X) := \det(A-BD^{-1}C)(\det D)^{-1}.$$
\end{definition}
If $B$ and or $C$ are zero, then $\Z_2^2\textrm{Ber}(X) := \det(A)(\det D)^{-1}.$

\section{Berezin integration on $\R^{2|1,1,1}$}\label{App:BerInt}
In this appendix, we follow Poncin \cite{Poncin:2016} closely making only the necessary changes to incorporate two coordinates of degree zero. The fundamental results are the same, namely, the theory of Berezin integration we use is mathematically well-founded. This is a priori not obvious. It is clear that as $\Z_2^2$-manifolds $\cM_2^{(1,1)} = \R^{2|1,1,1}$. As such we will change notation slightly to reflect that fact, and, in general, we are not working with any privileged class of coordinates or assign any physical meaning to them.  We will employ global coordinate systems of the form
 $$\{\underbrace{t}_{(0,0)}, ~ \underbrace{s}_{(0,0)}, ~ \underbrace{z}_{(1,1)}, ~\underbrace{\theta}_{(0,1)} ,~ \underbrace{\eta}_{(1,0)}  \}.$$
 \begin{definition}
 An \emph{integrable Berezin section} of $\R^{2|1,1,1}$ is a compactly supported Berezin section $\Bsigma = \mathrm{D}[t,s,z, \theta, \eta] \, \sigma(t,s,z,\theta, \eta)$ that does not contain the monomial $z \sigma(t,s)$.
 \end{definition}
 Using multi-index notation, we write
 $$\sigma(t,s,z, \theta, \eta) = \sum z^\kappa \theta^\alpha \eta^\beta \, \sigma_{\beta \alpha \kappa}(t,s),$$
 where $\kappa \in \mathbb{N}$, and $\alpha, \beta \in \{0,1\}$. 
 \begin{definition}
 Let $\Bsigma$ be an integrable Berezin section of $\R^{2|1,1,1}$. Furthermore, let us assume that an orientation as been chosen on $\R^2$. Then the \emph{Berezin integral} of $\Bsigma$ is defined as
 $$\int_{\R^{2|1,1,1}} \Bsigma := \int_{\R^2}\rmd t \rmd s \, \sigma_{110}(t,s).$$
 \end{definition}
 From this definition, it is natural to define $\int\mathrm{D}[z] z^\kappa = 0$ for all $\kappa > 0$, and $\int\mathrm{D}[z] =1$.  Or better, we will chose a formal integration interval and write $\int_0^1 \rmd z =1$.  This must be treated as a formal expression and we consider $\rmd z$ to be of $\Z_2^2$-degree $(1,1)$ and similarly for the symbol $\int_0^1$. Thus, the expressions are consistent. Note that we now have to take care with the formal integration interval when changing coordinates. In particular, if $z' = \theta \eta\, \phi_{110}^z + z \,\phi_{001}^z(t,s) + z^2 \theta \eta\, \phi_{112}^z(t,s) + \cO(z^3)$, then
 $$\int\rmd z' = \int \rmd z \, \left( \frac{\partial z'}{\partial z}\right)=  \int \rmd z \, (  \phi_{001}^z(t,s) + 2 z \theta \eta \,\phi_{112}^z(t,s) + \cO(z^2)) = \int \rmd z\,  \phi_{001}^z(t,s).$$
 Thus, we have to agree that the formal integral interval $z' \in [0,1]$ must transform into the formal integral interval $z \in [0, (\phi_{001}^z)^{-1}]$. This will be important when proving that the Berezin integral is independent of the coordinate system used.
\begin{remark}
Note that on $\cM_2^{(1,1)}$ that the coordinate $z$ is  Lorentz scalar and that $\Z_2^2$-supertranslations are linear in $z$ with unit coefficient. Thus, in both cases, $\partial z'/ \partial z =1$ and so the formal integration limit is not affected. Of course, as illustrated above, this is not the case for general coordinate transformations. 
\end{remark} 

To show that the Berezin integral as given here and used in Subsection \ref{SubSec:Int} is well-defined, we need to show the following.
\begin{enumerate}
\item \label{List:1} For any function $\sigma$ the property of having no term of the form $z\sigma_z(t,s)$ is independent of the chosen coordinates,
\item \label{List:2}that the coordinate Berezin volume does not introduce a term of the from $z f(t,s)$ under changes of coordinates, and finally,
\item \label{List:3}that the definition of the Berezin integral is independent of the chosen coordinates. 
\end{enumerate}
We will write general coordinate transformations as
\begin{align*}
\phi^*t' = t'(t,s,\theta, \eta, z) & = \phi^t(t,s) + z  \theta \eta \, \phi^t_{111}(t,s) + \cO(z^2),\\
\phi^*s' = s'(t,s,\theta, \eta, z) & = \phi^s(t,s) + z \theta \eta \, \phi^s_{111}(t,s) + \cO(z^2),\\
\phi^*z' = z'(t,s,\theta, \eta, z) & =  z  \, \phi^z_{001}(t,s) + \theta \eta \, \phi^z_{110}(t,s)  + \cO(z^3),\\
\phi^*\theta' = \theta'(t,s,\theta, \eta, z) & = \theta \,\phi_{010}^\theta(t,s) + z \eta  \, \phi^\theta_{101}(t,s) + \cO(z^2),\\
\phi^*\eta' = \eta'(t,s,\theta, \eta, z) & =  \eta  \, \phi^\eta_{100}(t,s) + z \theta  \, \phi^\eta_{011}(t,s)  + \cO(z^2),
\end{align*}
remembering that such transformations must preserve  the $\Z_2^2$-degree and that both $\theta$ and $\eta$ are nilpotent. The subscripts denote the number of formal coordinates in reverse order. In general, we have a formal power series in $z$. Note that the transformation rule for $z$ is liable to spoil the definition of the Berezin integral.  We will denote the Jacobian of these coordinate changes as
$$J = \frac{\partial(t', s',z', \theta', \eta')}{\partial(t,s,z,\theta,\eta, )}.$$
Consider a general monomial in the formal coordinates
$$z'^\kappa  \theta'^\alpha  \eta'^\beta,$$ 
where $\kappa \in \mathbb{N}$ and $\alpha, \beta \in \{ 0,1 \}$. Assuming that $\kappa \neq 1$, then applying the change of coordinates to the monomial never produces a term like $z f(t,s)$. If $\kappa =1$ and both $\alpha$ and $\beta$ are not zero, we reach the same conclusion. Then only  the monomial $z'$ can pick up a term of the form $zf(t,s)$ under coordinate transformations. As there are terms in the transformation rules for $t$ and $s$ that contain no formal coordinates (invertibility shows that these terms must be present) we conclude that only a term like $z'\sigma'(t', s')$ transforms in a way to produce  terms of the form $z \sigma(t,s)$. Thus, we have established \eqref{List:1}.
\begin{proposition}\label{Prop:NoZTermSigma}
If  in some choice of coordinates a function $\sigma'(t', s', z',\theta' , \eta')$  does not contain  the monomial $z'\sigma'(t', s')$, then in any other coordinate system the function $\sigma(t, s, z, \theta , \eta)$ similarly  does not contain  the monomial $z\sigma_{z}(t, s)$.
\end{proposition}
We will say that a \emph{monomial of type} $z^\kappa  \theta^\alpha \eta^\beta$ is any monomial of the form $z^k  \theta^a \eta^a$, with $k \geq \kappa$, $a \geq \alpha$ and $b \geq \beta$. Poncin \cite[Theorem 3]{Poncin:2016} states the following.
\begin{proposition}\label{Prop:MonTypBer}
A monomial of the type $z^\kappa  \theta^\alpha \eta^\beta$ in an entry of $J$ induces  in $\Z_2^2\Ber(J)$ only terms of the same type.
\end{proposition}
We will not prove this here and defer the reader to Poncin.
\begin{proposition}
The $\Z_2^2$-Berezinian of the Jacobian $\Z_2^2\Ber(J)$, does not contain a term of the form $z f(t,s)$. 
\end{proposition}
\begin{proof}
We are looking for a term $z f(t,s)$, so we can drop terms in the $\Z_2^2$-Berezin of the form $z^k \theta^a \eta^b$ with $k \geq 2$, $a\geq 0$, $b \geq 0$, as well as terms  $\theta$ and $\eta$.  In light of Proposition \ref{Prop:MonTypBer}, we can drop terms in the coordinate transformations that generate terms exclusively of the type $z^2$,  $\theta$ and $\eta$. Thus, we need only consider
\begin{align*}
\phi^*t' = t'(t,s,\theta, \eta, z) & = \phi^t(t,s) + z^2 \phi^t_{002},\\
\phi^*s' = s'(t,s,\theta, \eta, z) & = \phi^s(t,s)+ z^2 \phi^s_{002},\\
\phi^*z' = z'(t,s,\theta, \eta, z) & =  z  \, \phi^z_{001}(t,s) ,\\
\phi^*\theta' = \theta'(t,s,\theta, \eta, z) & = \theta \,\phi_{010}^\theta(t,s) + z \eta  \, \phi^\theta_{101}(t,s) \\
\phi^*\eta' = \eta'(t,s,\theta, \eta, z) & =  \eta  \, \phi^\eta_{100}(t,s) + z \theta  \, \phi^\eta_{011}(t,s).
\end{align*}
Using  Proposition \ref{Prop:MonTypBer} again, we can drop terms in the Jacobi matrix that are $\cO(z^2)$ and  those that contain $\eta$ and $\theta$. Thus, we need only examine the following matrix,
$$\renewcommand\arraystretch{1.5}  \begin{pmatrix}
    \partial_t\phi^t & \partial_s\phi^t  & 2 z \phi^t_{002} & 0 &0   \\
   \partial_t\phi^s & \partial_s\phi^s  & 2 z \phi^s_{002} & 0&0 \\
     z\partial_t \phi^z_{001} & z\partial_s \phi^z_{001}  &  \phi^z_{001} &0 & 0\\
    0 & 0  &  0 & \phi^\theta_{010} & - z \phi^\theta_{101} \\
    0 & 0  &  0& - z \phi^{\eta}_{011}& \phi^\eta_{100} \\
\end{pmatrix}.
$$ 
From Definition \ref{Def:Z22Ber} it is clear that the determinant of the top-left block of the above matrix has no term linear in $z$, one can directly calculate this determinant. The determinant of the bottom-right block is $\phi^\theta_{010}\phi^\eta_{100}  - z^2 \, \phi^\theta_{101} \phi^\eta_{011}$. We need the inverse of this, which is a formal power series in $z$. It is well-known that $(1 -z^2)^{-1} =1 + z^2 + z^4 \cdots$ and so it is clear that the inverse of the determinant of the bottom-right block does not contain a term linear in $z$. Thus, $\Z_2^2\Ber(J)$ does not contain a $zf(t,s)$.\\
\end{proof}
\begin{theorem}
The Berezin integral of an integrable Berezinian section of $\R^{2|1,1,1}$ (once an orientation has been fixed) is well-defined. 
\end{theorem}
\begin{proof}
We need to show that the Berezin integral of an integrable Berezinian section is independent of the coordinate system. In particular,
\begin{align*}& \int \mathrm{D}[t,s,z,\theta, \eta] \, \sigma(t,s,z,\theta, \eta)\\ & =  \int\mathrm{D}[t,s,z,\theta, \eta] \, \Z_2^2\textrm{Ber}\left(\frac{\partial (t',s',z',\theta', \eta')}{\partial(t,s,z,\theta, \eta) } \right) \, \sigma'\big(t'(t,s, \theta, \eta), s'(t,s, \theta, \eta), z'(t,s, \theta, \eta), \theta'(t,s, \theta, \eta), \eta'(t,s, \theta, \eta) \big).
\end{align*}
Due to the definition of the Berezin integral, we can ignore terms of the type $z$ in the integrand. In effect, we can set $z=0$ in the integrand. So, we can restrict attention to functions of the form
$$\sigma' = \sum \big(\theta \eta \phi^z_{110}(t,s) \big)^\kappa\big(\theta \phi^\theta_{010}(t,s) \big)^\alpha \big(\eta \phi^\eta_{100}(t,s) \big)^\beta \, \sigma'_{\beta \alpha \kappa}(\phi^t(t,s), \phi^s(t,s)).$$
Following Proposition \ref{Prop:MonTypBer} we can drop terms of type $z$ when calculating the $\Z_2^2$-Berezinian. Thus, we can consider the following simplified Jacobian matrix
\[\renewcommand\arraystretch{1.5} 
  \begin{pmatrix}
    \partial_t \phi^t & \partial_s\phi^t  & \theta \eta \, \phi^t_{110} &  0 &0   \\
    \partial_t \phi^s & \partial_s\phi^s  & \theta \eta \, \phi^s_{110} & 0 &0   \\
    \theta\eta \, \partial_t \phi^z_{110} & \theta\eta \, \partial_s \phi^z_{110}  &  \phi^z_{001} & \eta\, \phi^z_{110} & \theta\, \phi^z_{110} \\
    \theta \, \partial_t \phi^\theta_{010} & \theta \, \partial_s \phi^\theta_{010}  &  \eta \, \phi^\theta_{101} & \phi^\theta_{010} & 0 \\
    \eta \, \partial_t \phi^\eta_{100} & \eta\, \partial_s \phi^\eta_{100}  &  \theta\, \phi^\eta_{011}& 0 & \phi^\eta_{100} 
\end{pmatrix}
\]
Direct calculation gives
$$\det(A-BD^{-1}C) =  \phi^z_{001} \, \det \begin{pmatrix} \partial_t \phi^t & \partial_s \phi^t \\ \partial_t \phi^s & \partial_s \phi^s \end{pmatrix},$$
and 
$$(\det D)^{-1} = \frac{1}{\phi^\theta_{010}  \phi^\eta_{100}}.$$
Then, remembering we have assumed integrability, 
\begin{align*} 
&\int \mathrm{D}[t,s,z,\theta, \eta] \, \frac{\phi^z_{001} }{\phi^\theta_{010}\phi^\eta_{100}}   \, \det \begin{pmatrix} \partial_t \phi^t & \partial_s \phi^t \\ \partial_t \phi^s & \partial_s \phi^s \end{pmatrix}  ~ \theta \, \phi^\theta_{010}   \, \eta\, \phi^\eta_{100} \, \sigma'_{110}(\phi^t(t,s), \phi^s(t,s))\\
&= \int_{\R^2} \rmd t \rmd s \, \left(\int\rmd z  \; \phi^z_{001}\right) ~ \det \begin{pmatrix} \partial_t \phi^t & \partial_s \phi^t \\ \partial_t \phi^s & \partial_s \phi^s \end{pmatrix}~ \sigma'_{110}(\phi^t(t,s), \phi^s(t,s))\\
& = \int_{\R^2} \rmd t \rmd s  ~ \det \begin{pmatrix} \partial_t \phi^t & \partial_s \phi^t \\ \partial_t \phi^s & \partial_s \phi^s \end{pmatrix}~ \sigma'_{110}(\phi^t(t,s), \phi^s(t,s))\\
&= \int_{\R^2} \rmd t \rmd s  \, \sigma_{110}(t, s).
\end{align*}
\end{proof}
\end{document}